\documentclass[12pt]{article}

\usepackage[T1]{fontenc}
\usepackage[sc]{mathpazo}
\usepackage{amsmath}
\usepackage{amssymb}
\usepackage{enumerate}
\usepackage{amsthm}% theorems and proofs
\usepackage{amsfonts,mathrsfs}
\usepackage{hyperref}
\hypersetup{pdfpagemode=UseNone}

\newcommand{\tinyspace}{\mspace{1mu}}

\newcommand{\abs}[1]{\left\lvert\tinyspace #1 \tinyspace\right\rvert}

\newcommand{\setft}[1]{\mathrm{#1}}
\newcommand{\lin}[1]{\setft{L}\left(#1\right)}
\newcommand{\density}[1]{\setft{D}\left(#1\right)}

\newcommand{\trans}[1]{\setft{T}\left(#1\right)}

\def\complex{\mathbb{C}}

\def\natural{\mathbb{N}}

\def\I{\mathbb{1}}

\newenvironment{mylist}[1]{\begin{list}{}{
    \setlength{\leftmargin}{#1}
    \setlength{\rightmargin}{0mm}
    \setlength{\labelsep}{2mm}
    \setlength{\labelwidth}{8mm}
    \setlength{\itemsep}{0mm}}}
    {\end{list}}

\def\ot{\otimes}

\newcommand{\inner}[2]{\langle #1 , #2\rangle}

\newcommand{\out}[2]{| #1\rangle\langle #2 |}

\newcommand{\Innerm}[3]{\left\langle #1 \left| #2 \right| #3 \right\rangle}
\newcommand{\defeq}{\stackrel{\smash{\textnormal{\tiny def}}}{=}}

\newcommand{\Pa}[1]{\left(#1\right)}

\newcommand{\Br}[1]{\left[#1\right]}
\newcommand{\set}[1]{\{#1\}}
\newcommand{\Set}[1]{\left\{#1\right\}}

\newcommand{\ket}[1]{|#1\rangle}
\newcommand{\Ket}[1]{\left|#1\right\rangle}

\DeclareMathOperator{\trace}{Tr}

\newcommand{\Ptr}[2]{\trace_{#1}\Pa{#2}}

\newcommand{\Tr}[1]{\Ptr{}{#1}}

\def\cH{\mathcal{H}}

\def\rM{\mathrm{M}}
\def\rS{\mathrm{S}}

\newtheorem{thrm}{Theorem}[section]
\newtheorem{lem}[thrm]{Lemma}
\newtheorem{prop}[thrm]{Proposition}
\newtheorem{cor}[thrm]{Corollary}
\theoremstyle{definition}

\newtheorem{remark}[thrm]{Remark}

\numberwithin{equation}{section}

\newcounter{questionnumber}

\begin{document}

%=============================================================================%
\title{\Large\bf Unifying Treatment of Discord via Relative Entropy}
%=============================================================================%

\author{Lin Zhang\footnote{E-mail: godyalin@163.com}\\[1mm]
  {\it\small Institute of Mathematics, Hangzhou Dianzi University, Hangzhou 310018, P.R.~China}\\[1mm]
  Shao-Ming Fei\footnote{E-mail: feishm@mail.cnu.edu.cn}\\[1mm]
{\it\small School of Mathematical Sciences, Capital Normal
University, Beijing 100048, P.R.~China}\\
Jun Zhu\\
{\it\small Institute of Mathematics, Hangzhou Dianzi
University, Hangzhou 310018, P.R.~China}}

\date{}
\maketitle \mbox{}\hrule\mbox\\
\begin{abstract}

A new form of zero-discord state via Petz's monotonicity condition
on relative entropy with equality has been derived systematically. A
generalization of symmetric zero-discord states is presented and the
related physical implications are discussed.

\end{abstract}
\mbox{}\hrule\mbox\\

%=========================================================================%
\section{Introduction}
%=========================================================================%

Relative entropy are powerful tools in quantum information theory
\cite{Ohya}. It has a monotonicity property under a certain class of
quantum channels and the condition of equality is an interesting and
important subject. It is Petz who first studied the equality
condition of monotonicity of relative entropy \cite{Petz1,Petz2}.
Later Ruskai obtained similar result in terms of another elegant
approach \cite{Ruskai}. The most general equality condition along
with this line are recently reviewed in \cite{Hiai}.

In this note we will make use of the most general equality condition
for relative entropy to find the specific form of states which
satisfy the zero-discord condition (see details below).

Let $\cH$ denote an $N$-dimensional complex Hilbert space. A
\emph{state} $\rho$ on $\cH$ is a positive semi-definite operator of
trace one. We denote $\density{\cH}$ the set of all the density
matrices acting on $\cH$. If $\rho = \sum_k\lambda_k\out{u_k}{u_k}$
is the spectral decomposition of $\rho$, with $\lambda_k$ and
$|u_k\rangle$ the eigenvalues and eigenvectors respectively, then
the \emph{support} of $\rho$ is defined by
$$
\mathrm{supp}(\rho) \defeq \mathrm{span}\set{\ket{u_k} :
\lambda_k>0},
$$
and the \emph{generalized inverse} $\rho^{-1}$ of $\rho$ is defined
by
$$
\rho^{-1} = \sum_{k:\lambda_k>0}\lambda^{-1}_k\out{u_k}{u_k}.
$$
The \emph{von Neumann entropy} $\rS(\rho)$ of $\rho$ is defined by
$$
\rS(\rho) \defeq - \Tr{\rho\log\rho},
$$
which quantifies information encoded in the quantum state $\rho$. If
$\sigma$ is also a quantum state on $\cH$, then the \emph{relative
entropy} \cite{Ohya} between $\rho$ and $\sigma$ is defined by
$$
\rS(\rho||\sigma) \defeq \left\{\begin{array}{cl}
                             \Tr{\rho(\log\rho -
\log\sigma)}, & \text{if}\ \mathrm{supp}(\rho) \subseteq
\mathrm{supp}(\sigma), \\
                             +\infty, & \text{otherwise}.
                           \end{array}
\right.
$$

Let $\lin{\cH}$ be the set of all linear operators on $\cH$. If $X,
Y \in \lin{\cH}$, then $\inner{X}{Y} = \Tr{X^{\dagger}Y}$ defines
the \emph{Hilbert-Schmidt inner product} on $\lin{\cH}$. Let
$\trans{\cH}$ denote the set of all linear super-operators from
$\lin{\cH}$ to itself. $\Lambda\in \trans{\cH}$ is said to be a
\emph{completely positive super-operator} if for each $k \in
\natural$,
$$
\Lambda\ot \I_{\rM_{k}(\complex)}: \lin{\cH} \ot \rM_{k}(\complex)
\to \lin{\cH}\ot \rM_{k}(\complex)
$$
is positive, where $\rM_{k}(\complex)$ is the set of all $k\times k$
complex matrices. It follows from Choi's theorem \cite{Choi} that
every completely positive super-operator $\Lambda$ has a Kraus
representation
$$
\Lambda = \sum_{\mu}\mathrm{Ad}_{M_{\mu}},
$$
that is, for every $X\in \lin{\cH}$,  $\Lambda (X) = \sum_{\mu}M_\mu
XM_\mu^\dagger$, where $\set{M_\mu}\subseteq \lin{\cH}$,
$M_\mu^\dagger$ is the adjoint operator of $M_\mu$. It is clear that
for the super-operator $\Lambda$, there is \emph{adjoint
super-operator} $\Lambda^{\dagger}\in\trans{\cH}$ such that for $A,
B\in\lin{\cH}$, $\inner{\Lambda(A)}{B} =
\inner{A}{\Lambda^{\dagger}(B)}$. Moreover, $\Lambda$ is a
completely positive super-operator if and only if $\Lambda^\dagger$
is also a completely positive super-operator. A \emph{quantum
channel} is just a trace-preserving completely positive
super-operator $\Phi$. If $\Phi$ is also unit-preserving, then it is
called \emph{unital quantum channel}.

It has been reviewed in \cite{Hiai} that,
\begin{lem}\label{lem:Hiai}
Let $\rho,\sigma\in\density{\cH}$, $\Phi\in\trans{\cH}$ be a quantum
channel. If $\mathrm{supp}(\rho) \subseteq \mathrm{supp}(\sigma)$,
then $\rS(\Phi(\rho)||\Phi(\sigma)) \leqslant \rS(\rho||\sigma)$;
moreover
\begin{eqnarray*}
\rS(\Phi(\rho)||\Phi(\sigma)) = \rS(\rho||\sigma) \quad \mbox{if and
only if}\quad \Phi^\dagger_\sigma \circ \Phi(\rho) = \rho,
\end{eqnarray*}
where $\Phi^\dagger_\sigma = \mathrm{Ad}_{\sigma^{1/2}} \circ
\Phi^\dagger \circ \mathrm{Ad}_{\Phi(\sigma)^{-1/2}}$.
\end{lem}

Moreover, for a tripartite state, one has \cite{Linden,Cadney},
\begin{lem}\label{bi-SSA}
Let $\rho_{ABC}\in\density{\cH_A\ot\cH_B\ot\cH_C}$ for which strong
subadditivity is saturated for both triples $ABC, BAC$. Then
$\rho_{ABC}$ must have the following form:
$$
\rho_{ABC} = \bigoplus_{i,j}p_{ij} \rho_{a_i^L}^{(i)} \ot
\rho_{a_i^Rb_j^L}^{(ij)} \ot \rho_{b_j^R}^{(j)} \ot \rho_C^{(k)},
$$
where $k$ is a function only of $i,j$ in the sense that
$$
k = k(i,j) = k_1(i) = k_2(j) \quad \mbox{whenever} \quad p_{ij}>0.
$$
In particular, $k$ need only be defined where $p_{ij}>0$ so that it
is not necessarily constant. By collecting the terms of equivalent
$k$ we can write
$$
\rho_{ABC} = \bigoplus_k p_k \rho_{AB}^{(k)} \ot \rho_C^{(k)},
$$
where
$$
p_k\rho_{AB}^{(k)} = \sum_{i,j;k(i,j) = k} p_{ij} \rho_{a_i^L}^{(i)}
\ot \rho_{a_i^Rb_j^L}^{(ij)} \ot \rho_{b_j^R}^{(j)}.
$$
\end{lem}

%=========================================================================%
\section{Quantum discord}
%=========================================================================%

Consider a bipartite system $AB$ composed of subsystems $A$ and $B$.
Let $\rho_{AB}$ be the density operator of $AB$, and $\rho_A$ and
$\rho_B$ the reduced density operators. The total correlation
between the systems $A$ and $B$ is measured by the \emph{quantum
mutual information}
$$
I\Pa{\rho_{AB}} = \rS\Pa{\rho_A}  - \rS\Pa{\rho_A\big|\rho_B},
$$
where
$$
\rS\Pa{\rho_A\big|\rho_B} = \rS\Pa{\rho_{AB}} - \rS\Pa{\rho_B}
$$
is the entropy of $A$ conditional on $B$. The conditional entropy
can also be introduced by a measurement-based approach. Consider a
measurement locally performed on $B$, which can be described by a
set of projectors $\Pi_B =
\Set{\Pi_{B,\mu}}=\Set{\out{b_\mu}{b_\mu}}$. The state of the
quantum system, conditioned on the measurement of the outcome
labeled by $\mu$, is
$$
\rho_{AB,\mu} = \frac{1}{p_{B,\mu}}\Pa{\I_A \ot \Pi_{B,\mu}}
\rho_{AB} \Pa{\I_A \ot \Pi_{B,\mu}},
$$
where
$$
p_{B,\mu} = \Tr{(\I_A \ot \Pi_{B,\mu}) \rho_{AB} (\I_A \ot
\Pi_{B,\mu})} = \Innerm{b_\mu}{\rho_{AB}}{b_\mu}>0
$$
denotes the probability of obtaining the outcome $\mu$, and $\I_A$
denotes the identity operator for $A$. The conditional density
operator $\rho_{AB,\mu}$ allows for the following alternative
definition of the conditional entropy:
$$
\rS(\rho_{AB}|\Set{\Pi_{B,\mu}}) = \sum_\mu p_{B,\mu}
\rS(\rho_{AB,\mu}) = \sum_\mu p_{B,\mu} \rS(\rho_{A,\mu}),
$$
where $\rho_{A,\mu} = \Ptr{B}{\rho_{AB,\mu}} = (1/p_{B,\mu})
\Innerm{b_\mu}{\rho_{AB}}{b_\mu}$. Therefore, the quantum mutual
information can also be defined by
$$
I(\rho_{AB}|\Set{\Pi_{B,\mu}}) = \rS(\rho_A) -
\rS(\rho_{AB}|\Set{\Pi_{B,\mu}}).
$$
The quantities $I(\rho_{AB})$ and $I(\rho_{AB}|\Set{\Pi_{B,\mu}})$
are classically equivalent but distinct in the quantum case.

The one-sided \emph{quantum discord} is defined by:
$$
D_B(\rho_{AB}) = \inf_{\Pi_B} \Set{I(\rho_{AB}) -
I(\rho_{AB}|\Pi_B)}.
$$
If we denote the nonselective von Neumann measurement performed on
$B$ by
$$
\Pi_B(\rho_{AB}) = \sum_\mu (\I_A \ot \Pi_{B,\mu}) \rho_{AB} (\I_A
\ot \Pi_{B,\mu}) = \sum_\mu p_{B,\mu} \rho_{A,\mu} \ot
\out{b_\mu}{b_\mu},
$$
then the quantum discord can be written alternatively as
\begin{eqnarray*}
D_B(\rho_{AB}) &=& \inf_{\Pi_B} \Set{\rS(\rho_{AB}||\rho_A \ot
\rho_B) - \rS(\Pi_B(\rho_{AB})||\rho_A \ot
\Pi_B(\rho_B))}\\
&=& \inf_{\Pi_B} \Set{\rS(\rho_{AB}||\Pi_B(\rho_{AB})) -
\rS(\rho_B||\Pi_B(\rho_B))}.
\end{eqnarray*}
Apparently, $D_B(\rho_{AB})\geqslant0$ from Lemma~\ref{lem:Hiai}.

The \emph{symmetric quantum discord} of $\rho_{AB}$ is defined by
\cite{Rulli},
\begin{eqnarray*}
D(\rho_{AB}) =\inf_{\Pi_A \ot \Pi_B}
\Set{\rS(\rho_{AB}||\Pi_A\ot\Pi_B(\rho_{AB})) -
\rS(\rho_A||\Pi_A(\rho_A)) - \rS(\rho_B||\Pi_B(\rho_B))}.
\end{eqnarray*}
For the symmetric quantum discord of $\rho_{AB}$, one still has that
\begin{eqnarray}\label{ast}
D(\rho_{AB}) = \inf_{\Pi_A\ot\Pi_B} \Set{\rS(\rho_{AB}||\rho_A \ot
\rho_B) - \rS(\Pi_A \ot \Pi_B(\rho_{AB})||\Pi_A\ot\Pi_B(\rho_A \ot
\rho_B))}.
\end{eqnarray}

The symmetric quantum discord of $\rho_{A_1\ldots A_N}$ for
$N$-partite systems are defined by:
\begin{eqnarray*}
&&D(\rho_{A_1\ldots A_N}) \\
&&= \inf_{\Pi_{A_1} \ot \cdots \ot \Pi_{A_N}}
\Set{\rS(\rho_{A_1\ldots A_N}||\Pi_{A_1} \ot \cdots \ot
\Pi_{A_N}(\rho_{A_1\ldots A_N}))-\sum_{i=1}^N\rS(\rho_{A_i}||\Pi_{A_i}(\rho_{A_i}))}\\
&&=\inf_{\Pi_{A_1} \ot \cdots \ot \Pi_{A_N}} \{\rS(\rho_{A_1\ldots
A_N}||\rho_{A_1}\ot\cdots\ot\rho_{A_N})\\
&&- \rS(\Pi_{A_1} \ot \cdots \ot \Pi_{A_N}(\rho_{A_1\ldots
A_N})||\Pi_{A_1} \ot \cdots \ot
\Pi_{A_N}(\rho_{A_1}\ot\cdots\ot\rho_{A_N}))\},
\end{eqnarray*}
which is non-negative, $D(\rho_{A_1\ldots A_N}) \geqslant 0$.

The following theorem describes the structure of symmetric
zero-discord states:
\begin{thrm}
$D(\rho_{AB}) = 0$ if and only if
$$
\rho_{AB} = \sum_{\mu,\nu} \frac{p_{AB,\mu\nu}}{p_{A,\mu}
p_{B,\nu}}\sqrt{\rho_A}\Pi_{A,\mu}\sqrt{\rho_A} \ot
\sqrt{\rho_B}\Pi_{B,\nu}\sqrt{\rho_B}
$$
for both von Neumann measurements $\Pi_A = \Set{\Pi_{A,\mu}}$ and
$\Pi_B = \Set{\Pi_{B,\nu}}$, where
$$
p_{A,\mu} = \Tr{\Pi_{A,\mu}\rho_A},\quad p_{B,\nu} =
\Tr{\Pi_{B,\nu}\rho_B},\quad p_{AB,\mu\nu} = \Tr{\Pi_{A,\mu} \ot
\Pi_{B,\nu} \rho_{AB}}.
$$
\end{thrm}

\begin{proof}
Clearly, $\mathrm{supp}\Pa{\rho_{AB}} \subseteq
\mathrm{supp}\Pa{\rho_A} \ot \mathrm{supp}\Pa{\rho_B} =
\mathrm{supp}\Pa{\rho_A \ot \rho_B}$ \cite{Renner}. Since
$D(\rho_{AB})=0$, from Eq.~\eqref{ast}, it follows that there exist
two von Neumann measurement $\Pi_A = \Set{\Pi_{A,\mu}}$ and $\Pi_B =
\Set{\Pi_{B,\nu}}$ such that
$$
\rS(\Pi_A \ot \Pi_B(\rho_{AB})||\Pi_A\ot\Pi_B(\rho_A \ot \rho_B)) =
\rS(\rho_{AB}||\rho_A \ot \rho_B).
$$
Assume that $\sigma = \rho_A \ot \rho_B, \Phi = \Pi_A \ot \Pi_B$ in
Lemma~\ref{lem:Hiai}. Therefore $D(\rho_{AB}) = 0$ if and only if
$$
\rS(\Pi_A \ot \Pi_B(\rho_{AB})||\Pi_A\ot\Pi_B(\rho_A \ot \rho_B)) =
\rS(\rho_{AB}||\rho_A \ot \rho_B).
$$
Namely
$$ \rho_{AB} = \Phi^\dagger_\sigma \circ\Phi(\rho_{AB}) =
((\Pi^\dagger_{A,\rho_A}\circ\Pi_A) \ot
(\Pi^\dagger_{B,\rho_B}\circ\Pi_B))(\rho_{AB})
$$
Therefore
$$
\rho_{AB} = \sum_{\mu,\nu} \frac{p_{AB,\mu\nu}}{p_{A,\mu}
p_{B,\nu}}\sqrt{\rho_A}\Pi_{A,\mu}\sqrt{\rho_A} \ot
\sqrt{\rho_B}\Pi_{B,\nu}\sqrt{\rho_B}.
$$
\end{proof}

Accordingly we have
\begin{cor}
$D_B(\rho_{AB}) = 0$ if and only if
\begin{equation}\label{c1}
\rho_{AB} = \sum_\mu \rho_{A,\mu} \ot
\sqrt{\rho_B}\Pi_{B,\mu}\sqrt{\rho_B}
\end{equation}
for some von Neumann measurement $\Pi_B = \Set{\Pi_{B,\mu}}$, where
$$
\rho_{A,\mu} = \frac{1}{p_{B,\mu}}\Ptr{B}{\I_A \ot
\Pi_{B,\mu}\rho_{AB}},\quad p_{B,\mu} = \Tr{\Pi_{B,\mu}\rho_B}.
$$
\end{cor}

\begin{remark}
Suppose that the von Neumann measurement in Eq.~\eqref{c1} is $\Pi_B
= \Set{\Pi_{B,\mu}} = \Set{\out{b_\mu}{b_\mu}}$. Then we can assert
that $\Ket{b_\mu}$ is the eigenvectors of $\rho_B$. This can be seen
as follows. From Eq.~\eqref{c1} it follows that
\begin{equation}\label{r1}
\Pi_B(\rho_{AB}) = \sum_\mu \rho_{A,\mu} \ot
\Pi_B(\sqrt{\rho_B}\Pi_{B,\mu}\sqrt{\rho_B}).
\end{equation}
Actually,
\begin{equation}\label{r2}
\Pi_B(\rho_{AB}) = \sum_\mu (\I_A \ot \Pi_{B,\mu}) \rho_{AB} (\I_A
\ot \Pi_{B,\mu}) = \sum_\mu p_{B,\mu} \rho_{A,\mu} \ot \Pi_{B,\mu}.
\end{equation}
From Eq.~\eqref{r1} and Eq.~\eqref{r2}, we have
$$
\Pi_B(\sqrt{\rho_B}\Pi_{B,\mu}\sqrt{\rho_B}) = p_{B,\mu}\Pi_{B,\mu},
$$
which implies that
\begin{eqnarray}
\left\{\begin{array}{cc}
  \Pi_{B,\mu}\sqrt{\rho_B}\Pi_{B,\nu}\sqrt{\rho_B}\Pi_{B,\mu} = 0, & \mbox{if}\quad \mu\neq \nu, \\[3mm]
  \Pi_{B,\mu}\sqrt{\rho_B}\Pi_{B,\mu}\sqrt{\rho_B}\Pi_{B,\mu} = p_{B,\mu}\Pi_{B,\mu}, &
  \mbox{otherwise}.
\end{array}\right.
\end{eqnarray}
That is,
\begin{eqnarray*}
\left\{\begin{array}{cc}
  \abs{\Innerm{b_\mu}{\sqrt{\rho_B}}{b_\nu}}^2 = 0 & \mbox{if}\quad \mu\neq \nu, \\
  \Innerm{b_\mu}{\sqrt{\rho_B}}{b_\mu} = \sqrt{p_{B,\mu}} = \sqrt{\Innerm{b_\mu}{\rho_B}{b_\mu}} &
  \mbox{otherwise}.
\end{array}\right.
\end{eqnarray*}
Thus we conclude that $\Set{\Ket{b_\mu}}$ is the eigenvectors of
$\rho_B$.
\end{remark}

For general multipartite case we have
\begin{cor}
$D(\rho_{A_1\ldots A_N}) = 0$ if and only if
$$
\rho_{A_1\ldots A_N} = \sum_{\mu_1,\ldots,\mu_N} \frac{p_{A_1\ldots
A_N,\mu_1\ldots \mu_N}}{p_{A_1,\mu_1}\cdots
p_{A_N,\mu_N}}\sqrt{\rho_{A_1}}\Pi_{A_1,\mu_1}\sqrt{\rho_{A_1}} \ot
\cdots \ot \sqrt{\rho_{A_N}}\Pi_{A_N,\mu_N}\sqrt{\rho_{A_N}}
$$
for $N$ von Neumann measurements $\Pi_{A_i} =
\Set{\Pi_{A_i,\mu_i}}$, where
$$
p_{A_i,\mu_i} = \Tr{\Pi_{A_i,\mu_i}\rho_{A_i}}(i = 1,\ldots,N),\quad
p_{A_1\ldots A_N,\mu_1\ldots \mu_N}= \Tr{\Pi_{A_1,\mu_1} \ot \cdots
\ot \Pi_{A_N,\mu_N} \rho_{A_1\ldots A_N}}.
$$
\end{cor}

In order to obtain a connection with strong subadditivity of quantum
entropy \cite{Datta}, we associate each von Neumann measurement
$\Pi_B = \Set{\Pi_{B,\mu}}$ with a system $C$ as follows:
\begin{eqnarray}\label{d1}
\sigma_{ABC} = V\rho_{AB}V^\dagger = \sum_{\mu,\nu} (\I_A \ot
\Pi_{B,\mu})\rho_{AB} (\I_A \ot \Pi_{B,\nu})\ot \out{\mu}{\nu}_C,
\end{eqnarray}
where
$$
V\ket{\psi_B} \defeq \sum_\mu \Pi_{B,\mu}\ket{\psi_B} \ot
\ket{\mu}_C
$$
is an isometry from $B$ to $BC$. From Eq.~\eqref{d1} we have
\begin{eqnarray*}
\sigma_{AB} &=& \Ptr{C}{V\rho_{AB}V^\dagger} = \Pi_B\Pa{\rho_{AB}} = \sum_\mu p_{B,\mu} \rho_{A,\mu} \ot \Pi_{B,\mu},\\
\sigma_{BC} &=& \Ptr{A}{V\rho_{AB}V^\dagger} = \sum_{\mu,\nu}
\Pi_{B,\mu}\rho_B\Pi_{B,\nu}\ot \out{\mu}{\nu}_C,\\
\sigma_B &=& \sum_\mu p_{B,\mu}\Pi_{B,\mu},
\end{eqnarray*}
where $p_{B,\mu} = \Tr{\rho_B\Pi_{B,\mu}}$. This implies that the
conditional mutual information between $A$ and $C$ conditioned on
$B$ is
\begin{eqnarray*}
I(A;C|B)_\sigma &\defeq& \rS(\sigma_{AB}) + \rS(\sigma_{BC}) -
\rS(\sigma_{ABC}) - \rS(\sigma_B)\\
&=& \sum_\mu p_{B,\mu}\rS(\rho_{A,\mu}) + \rS(\rho_B) - \rS(\rho_{AB})\\
&=& \rS(\rho_{AB}||\rho_A \ot \rho_B) - \rS(\Pi_B(\rho_{AB})||\rho_A
\ot \Pi_B(\rho_B)).
\end{eqnarray*}

Similarly we have
$$
I(A;B|C)_\sigma = \rS\Pa{\rho_{AB}||\rho_A \ot \rho_B} -
\rS(\Pi_B(\rho_{AB})||\rho_A \ot \Pi_B(\rho_B)).
$$
That is,
\begin{eqnarray}\label{DoubleSSA}
I(A;C|B)_\sigma = I(A;B|C)_\sigma = \rS(\rho_{AB}||\rho_A \ot
\rho_B) - \rS(\Pi_B(\rho_{AB})||\rho_A \ot \Pi_B(\rho_B)).
\end{eqnarray}
If Eq.~(\ref{DoubleSSA}) vanishes for some von Neumann measurement
$\Pi_B = \Set{\Pi_{B,\mu}}$, $I(A;C|B)_\sigma = I(A;B|C)_\sigma =
0$, then from Lemma~\ref{bi-SSA}(i),
$$
\sigma_{ABC} = \bigoplus_k p_k \sigma^{(k)}_A \ot \sigma^{(k)}_{BC}.
$$
If $D_B(\rho_{AB}) = \rS(\rho_{AB}||\rho_A \ot \rho_B) -
\rS(\Pi_B(\rho_{AB})||\rho_A \ot \Pi_B(\rho_B))$ for some von
Neumann measurement $\Pi_B$, then
$$
D_B(\rho_{AB}) = I(A;B|C)_\sigma.
$$
There exists a famous protocol---state redistribution---which gives
an operational interpretation of conditional mutual information
$I(A;B|C)_\sigma$ \cite{Devetak,Yard}. This amounts to give
implicitly an operational interpretation of quantum discord
\cite{Madhok,Cavalcanti}.

%==========================================================================%
\section{A generalization of zero-discord states}
%==========================================================================%

Denote
\begin{eqnarray*}
\Omega^0_A &\defeq& \Set{\rho_{AB}\in\density{\cH_A\ot\cH_B}:
D_A(\rho_{AB}) = 0},\\
\Omega^0 &\defeq& \Set{\rho_{AB}\in\density{\cH_A\ot\cH_B}:
D(\rho_{AB}) = 0}.
\end{eqnarray*}
Suppose $\rho_{AB}\in\density{\cH_A\ot\cH_B}$, with two marginal
density matrices are $\rho_A = \Ptr{B}{\rho_{AB}}$ and $\rho_B =
\Ptr{A}{\rho_{AB}}$, respectively. A sufficient condition for
zero-discord states has been derived in \cite{Ferraro}: if
$\rho_{AB}\in\Omega^0_A$, then $\Br{\rho_{AB},\rho_A\ot\I_B}=0$.

A characterization of $\Br{\rho_{AB},\rho_A\ot\I_B}=0$ is obtained
in \cite{Cesar}, $\Br{\rho_{AB},\rho_A\ot\I_B}=0$ if and only if
$\rho_{AB} = \Pi_A(\rho_{AB})$, where $\Pi_A = \Set{\Pi_{A,\mu}}$ is
some positive valued measurement for which each projector
$\Pi_{A,\mu}$ is of any rank. That is,
$$
\rho_{AB} = \sum_\mu (\Pi_{A,\mu}\ot\I_B) \rho_{AB}
(\Pi_{A,\mu}\ot\I_B).
$$

States $\rho_{AB}$ such that $\Br{\rho_{AB},\rho_A\ot\I_B}=0$ are
called \emph{lazy ones} with particular physical interpretations
\cite{Cesar}. Consider general evolution of the state in a
finite-dimensional composite system $AB$:
$$
\Br{\frac{d}{dt}\rho_{AB,t}}_{t=\tau} =
-\mathrm{i}\Br{H_{AB},\rho_{AB,\tau}},
$$
where the total Hamiltonian is $H_{AB}\equiv H_A\ot\I_B + \I_A\ot
H_B + H_{\mathrm{int}}$, which consists of the system, the
environment and the interaction Hamiltonians. Clearly, it is
required that $\Ptr{A}{H_{\mathrm{int}}} = \Ptr{B}{H_{\mathrm{int}}}
= 0$. For the system $A$, the change rate of the system entropy at a
time $\tau$ is given by \cite{Ferraro}:
\begin{eqnarray}
\Br{\frac{d}{dt}\rS(\rho_{A,t})}_{t=\tau} =
-\mathrm{i}\Tr{H_{\mathrm{int}}\Br{\rho_{AB,\tau},\log_2(\rho_{A,\tau})\ot\I_B}}.
\end{eqnarray}
Since the von Neumann entropy $\rS(\rho_X)$ of $\rho_X$ quantifies
the degree of decoherence of the system $X(=A,B)$, it follows that
the system entropy rates are independent of the $AB$ coupling if and
only if
$$
\Br{\frac{d}{dt}\rS(\rho_{A,t})}_{t=\tau} = 0,
$$
which is equivalent to the following expression:
$$
\Br{\rho_{AB,\tau},\log_2(\rho_{A,\tau})\ot\I_B} = 0
\Longleftrightarrow \Br{\rho_{AB,\tau},\rho_{A,\tau}\ot\I_B} = 0.
$$
In view of this point, the entropy of quantum systems can be
preserved from decoherence under any coupling between $A$ and $B$ if
and only if the composite system states are lazy ones.

From the symmetry with respect to $A$ and $B$, one has
\begin{eqnarray}
\Br{\frac{d}{dt}\rS(\rho_{B,t})}_{t=\tau} =
-\mathrm{i}\Tr{H_{\mathrm{int}}\Br{\rho_{AB,\tau},\I_A\ot\log(\rho_{B,\tau})}}.
\end{eqnarray}
Due to that
$$
\Br{\frac{d}{dt}I(\rho_{AB,t})}_{t=\tau} =
\Br{\frac{d}{dt}\rS(\rho_{A,t})}_{t=\tau} +
\Br{\frac{d}{dt}\rS(\rho_{B,t})}_{t=\tau} -
\Br{\frac{d}{dt}\rS(\rho_{AB,t})}_{t=\tau}
$$
and
$$
\Br{\frac{d}{dt}\rS(\rho_{AB,t})}_{t=\tau} = 0,
$$
we have further
\begin{eqnarray}\label{eq:mutualentropy}
\Br{\frac{d}{dt}I(\rho_{AB,t})}_{t=\tau} =
-\mathrm{i}\Tr{H_{\mathrm{int}}\Br{\rho_{AB,\tau},\log(\rho_{A,\tau}\ot\rho_{B,\tau})}}.
\end{eqnarray}

We can see from Eq.~(\ref{eq:mutualentropy}) that the
total correlation is preserved under any coupling between $A$ and
$B$ if and only if the mutual entropy rate of composite system $AB$
is zero:
$$
\Br{\frac{d}{dt}I(\rho_{AB,t})}_{t=\tau} = 0,
$$
which is equivalent to the following expression:
$$
\Br{\rho_{AB,\tau},\log(\rho_{A,\tau}\ot\rho_{B,\tau})} = 0
\Longleftrightarrow
\Br{\rho_{AB,\tau},\rho_{A,\tau}\ot\rho_{B,\tau}} = 0.
$$

Similarly, we have:
\begin{prop}
If $\rho_{AB}\in\Omega^0$, then $\Br{\rho_{AB},\rho_A\ot\rho_B}=0$.
\end{prop}
Moreover,
\begin{prop}
$\Br{\rho_{AB},\rho_A\ot\rho_B}=0$ if and only if $\rho_{AB} =
\Pi_A\ot\Pi_B(\rho_{AB})$, where $\Pi_X = \Set{\Pi_{X,\alpha}}$ are
some PVM for which each projector $\Pi_{X,\alpha}$, where
$(X,\alpha)= (A,\mu),(B,\nu)$, are of any rank. That is,
$$
\rho_{AB} = \sum_{\mu,\nu} (\Pi_{A,\mu}\ot\Pi_{B,\nu}) \rho_{AB}
(\Pi_{A,\mu}\ot\Pi_{B,\nu}).
$$
\end{prop}

\begin{proof}
Let the spectral decompositions of $\rho_{A,\tau}$ and
$\rho_{B,\tau}$ be
$$
\rho_{A,\tau} = \sum_\mu p_\mu \Pi_{A,\tau},\quad \rho_{B,\tau} =
\sum_\nu q_\nu \Pi_{B,\nu},
$$
respectively, where $\Set{\Pi_{A,\mu}}$ and $\Set{\Pi_{B,\nu}}$ are
the orthogonal projectors of any rank, such that $\set{p_\mu}$ and
$\set{q_\nu}$ are non-degenerate, respectively. Then
$\Set{\Pi_{A,\mu}\ot\Pi_{B,\nu}}$ are orthogonal eigen-projectors of
$\rho_{A,\tau}\ot\rho_{B,\tau}$. Since
$\Br{\rho_{AB},\rho_A\ot\rho_B}=0$ is equivalent to
$\Br{\rho_{AB},\Pi_{A,\mu}\ot\Pi_{B,\nu}}=0$ for all $\mu,\nu$, it
follows from $\sum_{\mu,\nu}\Pi_{A,\mu}\ot\Pi_{B,\nu} = \I_A\ot\I_B$
that
$$
\rho_{AB} = \sum_{\mu,\nu} (\Pi_{A,\mu}\ot\Pi_{B,\nu}) \rho_{AB}
(\Pi_{A,\mu}\ot\Pi_{B,\nu}).
$$
The converse follows from direct computation.
\end{proof}

Here the states $\rho_{AB}$ satisfying the condition
$\Br{\rho_{AB},\rho_A\ot\rho_B}=0$ are just the generalization of
zero-symmetric discord states and lazy states are the generalization
of zero discord states.

%==========================================================================%
\section{Conclusion}
%==========================================================================%
We have studied the well-known monotonicity inequality of relative
entropy under completely positive linear maps, by deriving some
properties of symmetric discord. A new form of zero-discord state
via Petz's monotonicity condition on relative entropy with equality
has been derived systematically. The results are generalized for the
zero-discord states.

There is a more interesting and challenging problem which can be
considered in the future study: What is a sufficient and necessary
condition for the vanishing conditional mutual entropy rates at a
time $\tau$:
$$
\Br{\frac{d}{dt}I(A:B|E)_\rho}_{t=\tau}=0,
$$
where $I(A:B|E)_\rho = \rS(\rho_{AE}) + \rS(\rho_{BE}) -
\rS(\rho_{ABE}) - \rS(\rho_E)$.

%=============================================================================%
\subsection*{Acknowledgement}
We thank F. Brand\~{a}o, M. Mosonyi, M. Piani, J. Rau and A. Winter
for valuable comments. This project is supported by Natural Science
Foundations of China (11171301, 10771191 and 10471124) and Natural
Science Foundation of Zhejiang Province of China (Y6090105).
%=============================================================================%

%=====================================================================================================================

\end{document}